\documentclass{llncs}

\usepackage{wrapfig}
\usepackage{amssymb,amsmath,txfonts,mathtools}
\usepackage{cmll}
\usepackage{stmaryrd}
\usepackage{todonotes}
\usepackage{mathpartir}
\usepackage{hyperref}
\usepackage[barr]{xy}
\usepackage{mdframed}
\usepackage{thmtools, thm-restate}
\usepackage{enumitem}

\let\mto\to
\let\to\relax
\newcommand{\to}{\rightarrow}

\newcommand{\interp}[1]{\llbracket #1 \rrbracket}

\date{}

\newcommand{\Four}[0]{\mathsf{4}}
\newcommand{\cat}[1]{\mathcal{#1}}

\newcommand{\Hom}[3]{\mathsf{Hom}_{\cat{#1}}(#2,#3)}
\newcommand{\limp}[0]{\multimap}
\newcommand{\dial}[1]{\mathsf{Dial_{#1}}(\mathsf{Sets})}

\newcommand{\obj}[1]{\mathsf{Obj}(#1)}

\newcommand{\id}[0]{\mathsf{id}}

\newenvironment{ATdefnblock}[3][]{ \framebox{\mbox{#2}} \quad #3 \\[0pt]}{}

\newcommand{\ATnt}[1]{\mathit{#1}}
\newcommand{\ATmv}[1]{\mathit{#1}}

\usepackage{cmll}
                 \usepackage{relsize}

\newcommand{\ATLLdrule}[4][]{{\displaystyle\frac{\begin{array}{l}#2\end{array}}{#3}\quad\ATLLdrulename{#4}}}

\newcommand{\ATLLpremise}[1]{ #1 \\}
\newenvironment{ATLLdefnblock}[3][]{ \framebox{\mbox{#2}} \quad #3 \\[0pt]}{}

\newcommand{\ATLLnt}[1]{\mathit{#1}}
\newcommand{\ATLLmv}[1]{\mathit{#1}}

\newcommand{\ATLLsym}[1]{#1}

\usepackage{cmll}
                 \usepackage{relsize}

\newcommand{\ATLLdruleCXXidName}[0]{\ATLLdrulename{C\_id}}
\newcommand{\ATLLdruleCXXid}[1]{\ATLLdrule[#1]{%
}{
 \Gamma  \vdash  \Gamma }{%
{\ATLLdruleCXXidName}{}%
}}

\newcommand{\ATLLdruleCXXcName}[0]{\ATLLdrulename{C\_c}}
\newcommand{\ATLLdruleCXXc}[1]{\ATLLdrule[#1]{%
\ATLLpremise{  \Gamma_{{\mathrm{1}}}  \vdash  \Gamma_{{\mathrm{2}}}   \quad   \Gamma_{{\mathrm{2}}}  \vdash  \Gamma_{{\mathrm{3}}}  }%
}{
 \Gamma_{{\mathrm{1}}}  \vdash  \Gamma_{{\mathrm{3}}} }{%
{\ATLLdruleCXXcName}{}%
}}

\newcommand{\ATLLdruleCXXaOneName}[0]{\ATLLdrulename{C\_a1}}
\newcommand{\ATLLdruleCXXaOne}[1]{\ATLLdrule[#1]{%
}{
  \ATLLsym{(}   \Gamma_{{\mathrm{1}}}  \circ  \Gamma_{{\mathrm{2}}}   \ATLLsym{)}  \circ  \Gamma_{{\mathrm{3}}}   \vdash   \Gamma_{{\mathrm{1}}}  \circ  \ATLLsym{(}   \Gamma_{{\mathrm{2}}}  \circ  \Gamma_{{\mathrm{3}}}   \ATLLsym{)}  }{%
{\ATLLdruleCXXaOneName}{}%
}}

\newcommand{\ATLLdruleCXXuOneName}[0]{\ATLLdrulename{C\_u1}}
\newcommand{\ATLLdruleCXXuOne}[1]{\ATLLdrule[#1]{%
}{
  \Gamma  \circ  \ATLLsym{*}   \vdash  \Gamma }{%
{\ATLLdruleCXXuOneName}{}%
}}

\newcommand{\ATLLdruleCXXuTwoName}[0]{\ATLLdrulename{C\_u2}}
\newcommand{\ATLLdruleCXXuTwo}[1]{\ATLLdrule[#1]{%
}{
  \ATLLsym{*}  \circ  \Gamma   \vdash  \Gamma }{%
{\ATLLdruleCXXuTwoName}{}%
}}

\newcommand{\ATLLdruleCXXeOneName}[0]{\ATLLdrulename{C\_e1}}
\newcommand{\ATLLdruleCXXeOne}[1]{\ATLLdrule[#1]{%
}{
  \Gamma ( \ATLLnt{A}  \ATLLsym{,}  \ATLLnt{B} )   \vdash   \Gamma ( \ATLLnt{B}  \ATLLsym{,}  \ATLLnt{A} )  }{%
{\ATLLdruleCXXeOneName}{}%
}}

\newcommand{\ATLLdruleCXXeTwoName}[0]{\ATLLdrulename{C\_e2}}
\newcommand{\ATLLdruleCXXeTwo}[1]{\ATLLdrule[#1]{%
}{
  \Gamma (  \ATLLnt{A}  \bullet  \ATLLnt{B}  )   \vdash   \Gamma (  \ATLLnt{B}  \bullet  \ATLLnt{A}  )  }{%
{\ATLLdruleCXXeTwoName}{}%
}}

\newcommand{\ATLLdruleCXXdOneName}[0]{\ATLLdrulename{C\_d1}}
\newcommand{\ATLLdruleCXXdOne}[1]{\ATLLdrule[#1]{%
}{
  \Gamma ( \ATLLnt{A}  \ATLLsym{;}  \ATLLsym{(}   \Delta_{{\mathrm{1}}}  \bullet  \Delta_{{\mathrm{2}}}   \ATLLsym{)} )   \vdash   \Gamma (  \ATLLsym{(}  \ATLLnt{A}  \ATLLsym{;}  \Delta_{{\mathrm{1}}}  \ATLLsym{)}  \bullet  \ATLLsym{(}  \ATLLnt{A}  \ATLLsym{;}  \Delta_{{\mathrm{2}}}  \ATLLsym{)}  )  }{%
{\ATLLdruleCXXdOneName}{}%
}}

\newcommand{\ATLLdruleCXXdTwoName}[0]{\ATLLdrulename{C\_d2}}
\newcommand{\ATLLdruleCXXdTwo}[1]{\ATLLdrule[#1]{%
}{
  \Gamma (  \ATLLsym{(}  \ATLLnt{A}  \ATLLsym{;}  \Delta_{{\mathrm{1}}}  \ATLLsym{)}  \bullet  \ATLLsym{(}  \ATLLnt{A}  \ATLLsym{;}  \Delta_{{\mathrm{2}}}  \ATLLsym{)}  )   \vdash   \Gamma ( \ATLLnt{A}  \ATLLsym{;}  \ATLLsym{(}   \Delta_{{\mathrm{1}}}  \bullet  \Delta_{{\mathrm{2}}}   \ATLLsym{)} )  }{%
{\ATLLdruleCXXdTwoName}{}%
}}

\newcommand{\ATLLdruleCXXdThreeName}[0]{\ATLLdrulename{C\_d3}}
\newcommand{\ATLLdruleCXXdThree}[1]{\ATLLdrule[#1]{%
}{
  \Gamma ( \ATLLnt{A}  \ATLLsym{,}  \ATLLsym{(}   \Delta_{{\mathrm{1}}}  \bullet  \Delta_{{\mathrm{2}}}   \ATLLsym{)} )   \vdash   \Gamma (  \ATLLsym{(}  \ATLLnt{A}  \ATLLsym{,}  \Delta_{{\mathrm{1}}}  \ATLLsym{)}  \bullet  \ATLLsym{(}  \ATLLnt{A}  \ATLLsym{,}  \Delta_{{\mathrm{2}}}  \ATLLsym{)}  )  }{%
{\ATLLdruleCXXdThreeName}{}%
}}

\newcommand{\ATLLdruleCXXdFourName}[0]{\ATLLdrulename{C\_d4}}
\newcommand{\ATLLdruleCXXdFour}[1]{\ATLLdrule[#1]{%
}{
  \Gamma (  \ATLLsym{(}  \ATLLnt{A}  \ATLLsym{,}  \Delta_{{\mathrm{1}}}  \ATLLsym{)}  \bullet  \ATLLsym{(}  \ATLLnt{A}  \ATLLsym{,}  \Delta_{{\mathrm{2}}}  \ATLLsym{)}  )   \vdash   \Gamma ( \ATLLnt{A}  \ATLLsym{,}  \ATLLsym{(}   \Delta_{{\mathrm{1}}}  \bullet  \Delta_{{\mathrm{2}}}   \ATLLsym{)} )  }{%
{\ATLLdruleCXXdFourName}{}%
}}

\newcommand{\ATLLdruleLXXvarName}[0]{\ATLLdrulename{L\_var}}
\newcommand{\ATLLdruleLXXvar}[1]{\ATLLdrule[#1]{%
}{
 \ATLLnt{B}  \vdash  \ATLLnt{B} }{%
{\ATLLdruleLXXvarName}{}%
}}

\newcommand{\ATLLdruleLXXnodeName}[0]{\ATLLdrulename{L\_node}}
\newcommand{\ATLLdruleLXXnode}[1]{\ATLLdrule[#1]{%
}{
 \ATLLsym{*}  \vdash  \ATLLmv{N} }{%
{\ATLLdruleLXXnodeName}{}%
}}

\newcommand{\ATLLdruleLXXCtxName}[0]{\ATLLdrulename{L\_Ctx}}
\newcommand{\ATLLdruleLXXCtx}[1]{\ATLLdrule[#1]{%
\ATLLpremise{  \Gamma_{{\mathrm{1}}}  \vdash  \Gamma_{{\mathrm{2}}}   \quad   \Gamma_{{\mathrm{2}}}  \vdash  \ATLLnt{A}  }%
}{
 \Gamma_{{\mathrm{1}}}  \vdash  \ATLLnt{A} }{%
{\ATLLdruleLXXCtxName}{}%
}}

\newcommand{\ATLLdruleLXXparaIName}[0]{\ATLLdrulename{L\_paraI}}
\newcommand{\ATLLdruleLXXparaI}[1]{\ATLLdrule[#1]{%
\ATLLpremise{  \Gamma  \vdash  \ATLLnt{A}   \quad   \Delta  \vdash  \ATLLnt{B}  }%
}{
 \Gamma  \ATLLsym{,}  \Delta  \vdash  \ATLLnt{A}  \odot  \ATLLnt{B} }{%
{\ATLLdruleLXXparaIName}{}%
}}

\newcommand{\ATLLdruleLXXchoiceIName}[0]{\ATLLdrulename{L\_choiceI}}
\newcommand{\ATLLdruleLXXchoiceI}[1]{\ATLLdrule[#1]{%
\ATLLpremise{  \Gamma  \vdash  \ATLLnt{A}   \quad   \Delta  \vdash  \ATLLnt{B}  }%
}{
  \Gamma  \bullet  \Delta   \vdash  \ATLLnt{A}  \sqcup  \ATLLnt{B} }{%
{\ATLLdruleLXXchoiceIName}{}%
}}

\newcommand{\ATLLdruleLXXseqIName}[0]{\ATLLdrulename{L\_seqI}}
\newcommand{\ATLLdruleLXXseqI}[1]{\ATLLdrule[#1]{%
\ATLLpremise{  \Gamma  \vdash  \ATLLnt{A}   \quad   \Delta  \vdash  \ATLLnt{B}  }%
}{
 \Gamma  \ATLLsym{;}  \Delta  \vdash  \ATLLnt{A}  \rhd  \ATLLnt{B} }{%
{\ATLLdruleLXXseqIName}{}%
}}

\newcommand{\ATLLdruleLXXparaEName}[0]{\ATLLdrulename{L\_paraE}}
\newcommand{\ATLLdruleLXXparaE}[1]{\ATLLdrule[#1]{%
\ATLLpremise{  \Gamma  \vdash  \ATLLnt{A}  \odot  \ATLLnt{B}   \quad    \Delta ( \ATLLnt{A}  \ATLLsym{,}  \ATLLnt{B} )   \vdash  \ATLLnt{C}  }%
}{
  \Delta ( \Gamma )   \vdash  \ATLLnt{C} }{%
{\ATLLdruleLXXparaEName}{}%
}}

\newcommand{\ATLLdruleLXXchoiceEName}[0]{\ATLLdrulename{L\_choiceE}}
\newcommand{\ATLLdruleLXXchoiceE}[1]{\ATLLdrule[#1]{%
\ATLLpremise{  \Gamma  \vdash  \ATLLnt{A}  \sqcup  \ATLLnt{B}   \quad    \Delta (  \ATLLnt{A}  \bullet  \ATLLnt{B}  )   \vdash  \ATLLnt{C}  }%
}{
  \Delta ( \Gamma )   \vdash  \ATLLnt{C} }{%
{\ATLLdruleLXXchoiceEName}{}%
}}

\newcommand{\ATLLdruleLXXseqEName}[0]{\ATLLdrulename{L\_seqE}}
\newcommand{\ATLLdruleLXXseqE}[1]{\ATLLdrule[#1]{%
\ATLLpremise{  \Gamma  \vdash  \ATLLnt{A}  \rhd  \ATLLnt{B}   \quad    \Delta ( \ATLLnt{A}  \ATLLsym{;}  \ATLLnt{B} )   \vdash  \ATLLnt{C}  }%
}{
  \Delta ( \Gamma )   \vdash  \ATLLnt{C} }{%
{\ATLLdruleLXXseqEName}{}%
}}

\newcommand{\ATLLdruleLXXlimpIName}[0]{\ATLLdrulename{L\_limpI}}
\newcommand{\ATLLdruleLXXlimpI}[1]{\ATLLdrule[#1]{%
\ATLLpremise{ \Gamma  \ATLLsym{,}  \ATLLnt{A}  \vdash  \ATLLnt{B} }%
}{
 \Gamma  \vdash  \ATLLnt{A}  \multimap  \ATLLnt{B} }{%
{\ATLLdruleLXXlimpIName}{}%
}}

\newcommand{\ATLLdruleLXXlimpEName}[0]{\ATLLdrulename{L\_limpE}}
\newcommand{\ATLLdruleLXXlimpE}[1]{\ATLLdrule[#1]{%
\ATLLpremise{  \Gamma  \vdash  \ATLLnt{A}  \multimap  \ATLLnt{B}   \quad   \Delta  \vdash  \ATLLnt{A}  }%
}{
 \Gamma  \ATLLsym{,}  \Delta  \vdash  \ATLLnt{B} }{%
{\ATLLdruleLXXlimpEName}{}%
}}

\newenvironment{OLLdefnblock}[3][]{ \framebox{\mbox{#2}} \quad #3 \\[0pt]}{}

\usepackage{cmll}
                 \usepackage{relsize}

\newenvironment{ATermsdefnblock}[3][]{ \framebox{\mbox{#2}} \quad #3 \\[0pt]}{}

\newcommand{\ATermsnt}[1]{\mathit{#1}}
\newcommand{\ATermsmv}[1]{\mathit{#1}}

\usepackage{cmll}
                 \usepackage{relsize}

\renewcommand{\ATLLdrule}[4][]{{\displaystyle\frac{\begin{array}{l}#2\end{array}}{#3}\,#4}}
\renewcommand{\ATLLdruleLXXvarName}{\text{id}}
\renewcommand{\ATLLdruleLXXnodeName}{\text{base}}
\renewcommand{\ATLLdruleLXXCtxName}{\text{CM}}
\renewcommand{\ATLLdruleLXXparaIName}{\odot_i}
\renewcommand{\ATLLdruleLXXchoiceIName}{\sqcup_i}
\renewcommand{\ATLLdruleLXXseqIName}{\rhd_i}
\renewcommand{\ATLLdruleLXXparaEName}{\odot_e}
\renewcommand{\ATLLdruleLXXchoiceEName}{\sqcup_e}
\renewcommand{\ATLLdruleLXXseqEName}{\rhd_e}
\renewcommand{\ATLLdruleLXXlimpIName}{\limp_i}
\renewcommand{\ATLLdruleLXXlimpEName}{\limp_e}
\renewcommand{\ATLLdruleCXXidName}{\text{id}}
\renewcommand{\ATLLdruleCXXcName}{\text{comp}}
\renewcommand{\ATLLdruleCXXaOneName}{\text{assoc}}
\renewcommand{\ATLLdruleCXXuOneName}{\text{unit}_1}
\renewcommand{\ATLLdruleCXXuTwoName}{\text{unit}_2}
\renewcommand{\ATLLdruleCXXeOneName}{\text{ex}_1}
\renewcommand{\ATLLdruleCXXeTwoName}{\text{ex}_2}
\renewcommand{\ATLLdruleCXXdOneName}{\text{dist}_1}
\renewcommand{\ATLLdruleCXXdTwoName}{\text{dist}_2}
\renewcommand{\ATLLdruleCXXdThreeName}{\text{dist}_3}
\renewcommand{\ATLLdruleCXXdFourName}{\text{dist}_4}

\begin{document}

\title{An Intuitionistic Linear Logical Semantics of SAND Attack Trees}

\author{Harley Eades III}
\institute{Computer Science\\Augusta University \\ \href{mailto:heades@augusta.edu}{harley.eades@gmail.com}}

\maketitle 

\begin{abstract}
  In this paper we introduce a new logical foundation of SAND attack
  trees in intuitionistic linear logic.  This new foundation is based
  on a new logic called the Attack Tree Linear Logic (ATLL).  Before
  introducing ATLL we given several new logical models of attack trees,
  the first, is a very basic model based in truth tables.  Then we lift
  this semantics into a semantics of attack trees based on lineales
  which introduces implication, but this can be further lifted into a
  dialectica model which ATLL is based.  One important feature of ATLL
  is that it supports full distributivity of sequential conjunction
  over choice.
\end{abstract}

\section{Introduction}
\label{sec:introduction}
Attack trees are a type of tree based graphical model used in
analyzing the threat potential of secure systems.  They were
popularized by Bruce Schneier \cite{Schneier:1999} at the NSA and they
were found to be affective at analyzing both physical and virtual
secure systems.  An attack tree represents all possible ways to
accomplish a particular attack on a system.  The root of the tree is
the over all goal, and then each child node represents a refinement of
the overall goal.  The leafs are particular attacks needed to
accomplish the goal of the tree.  An example attack tree for attacking
an ATM can be found in Figure~\ref{fig:atm-tree1}.

In this paper attack trees have three types of branching nodes:
and-node, or-nodes, and sequence-nodes.  And-nodes are depicted
graphically by drawing a horizontal line linking the children nodes.
This type of node represents a set of attacks that all must be
executed, but there is no particular order on which must be executed
first. Similarly, sequence-nodes are depicted by a horizontal arrow
linking the children nodes.  These nodes represent ordered execution
of attacks.  Finally, the remaining nodes are or-nodes which represent
a choice between attacks.  This formalization of attack trees is known
as SAND attack trees and were introduced by Jhawar
et. al~\cite{Jhawar:2015}.

The way in which we describe attack trees suggests that they should be
considered as describing a process in terms of smaller processes.
Then each leaf of an attack tree represents a process -- an attack --
that must be executed to reach the overall goal.  Branching nodes are
then operators of process algebra.  Thus, an attack tree is a process
tree representing an attack.

Until recently attack trees have been primarily a tool without a
theoretical foundation.  This has become worrisome, because several
projects have used them to asses the security of large scale systems.
Therefore, a leading question regarding attack trees is, what is a
mathematical model of attack trees?  There have been numerous proposed
answers to this question.  Some examples are propositional logic,
multisets, directed acyclic graphs, source sink graphs (or
parallel-series pomsets), Petri nets, and Markov processes.

Give that there are quite a few different models one can then ask, is
there a unifying foundation in common to each of these proposed
models?  Furthermore, can this unifying foundation be used to further
the field of attack trees and build new tools for conducting threat
analysis?  This paper contributes to the answer of these questions.
Each of the proposed models listed above have something in common.
They can all be modeled in some form of a symmetric monoidal category
\cite{Tzouvaras:1998,Brown:1991,Fiore:2013,FrancescoAlbasini2010}.
That is all well and good, but what can we gain from monoidal
categories?

Monoidal categories are a mathematical model of linear logic as
observed through the beautiful Curry-Howard-Lambek correspondence
\cite{Mellies:2009}.  In linear logic every hypothesis must be used
exactly once, and hence, if we view a hypothesis as a resource, then
this property can be stated as every resource must be consumed.  This
is an ideal setting for reasoning about processes like attack trees.

Multisets and Petri nets both capture the idea that the nodes of an
attack tree consist of the attack action and the state -- the
resource -- of the system being analyzed. As it turns out, linear
logic has been shown to be a logical foundation for multisets
\cite{Tzouvaras:1998} and Petri Nets \cite{Brown:1991}.  Thus, linear
logic has the ability to model the state as well as attack actions of
the goals of an attack tree.  

In this paper\footnote{This material is based upon work supported by
  the National Science Foundation CRII CISE Research Initiation grant,
  ``CRII:SHF: A New Foundation for Attack Trees Based on Monoidal
  Categories``, under Grant No. 1565557.} we introduce a new logical
foundation of SAND attack trees in intuitionistic linear logic.  This
new foundation is based on a new logic called the Attack Tree Linear
Logic (ATLL).

\begin{wrapfigure}{r}{0.5\textwidth}
  \vspace{-45px}
  \begin{center}
    \includegraphics[scale=0.210]{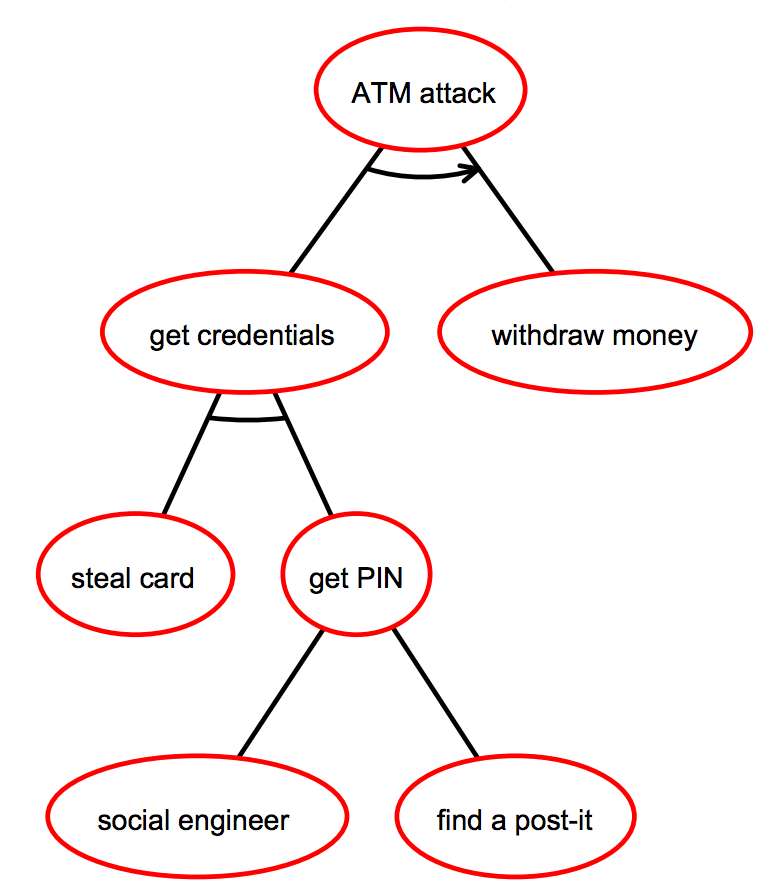}
  \end{center}  
  \label{fig:atm-tree1}
  \caption{Attack Tree for an ATM attack from Figure~1 on page 2 of Kordy et al.~\cite{Kordy2017}}
  \vspace{-55px}
\end{wrapfigure}
In ATLL (Section~\ref{sec:the_attack_tree_linear_logic_(atll)}) attack
trees are modeled as linear formulas where base attacks are atomic
formulas, and each branching node corresponds to a binary logical
connective.  Consider the attack tree for an ATM attack from
Figure~\ref{fig:atm-tree1}.  We can model this attack tree as a
formula in ATLL as follows:
\[
\begin{array}{lll}
  \ATLLnt{B_{{\mathrm{1}}}} := \text{``steal card''}\\
  \ATLLnt{B_{{\mathrm{2}}}} := \text{``social engineer''}\\
  \ATLLnt{B_{{\mathrm{3}}}} := \text{``find a post-it''}\\
  \ATLLnt{B_{{\mathrm{4}}}} := \text{``withdrawal money''}\\
  \ATLLnt{T_{{\mathrm{1}}}} := \ATLLsym{(}  \ATLLnt{B_{{\mathrm{1}}}}  \odot  \ATLLsym{(}  \ATLLnt{B_{{\mathrm{2}}}}  \sqcup  \ATLLnt{B_{{\mathrm{3}}}}  \ATLLsym{)}  \ATLLsym{)}  \rhd  \ATLLnt{B_{{\mathrm{4}}}}
\end{array}
\]
Each $\ATLLnt{B_{\ATLLmv{i}}}$ is an atomic formula, parallel conjunction -- and-nodes
-- of attack trees is denoted by $\ATLLnt{T_{{\mathrm{1}}}}  \odot  \ATLLnt{T_{{\mathrm{2}}}}$, choice between
attacks -- or-nodes -- by $\ATLLnt{T_{{\mathrm{1}}}}  \sqcup  \ATLLnt{T_{{\mathrm{2}}}}$, and sequential conjunction of
attacks -- sequence-nodes -- by $\ATLLnt{T_{{\mathrm{1}}}}  \rhd  \ATLLnt{T_{{\mathrm{2}}}}$.  Parallel conjunction
and choice are both symmetric, but sequential conjunction is not.  Now
that we can model attack trees as formulas we should be able to use
the logic to reason about attack trees.

Reasoning about attack trees corresponds to proving implications
between them.  In fact, every equation from Jhawar et al.'s
work on attack trees with sequential conjunction \cite{Jhawar:2015}
can be proven as an implication in ATLL.  Consider a second attack
tree from Figure~\ref{fig:atm-tree2}.
\begin{figure}
  \begin{center}
    \includegraphics[scale=0.35]{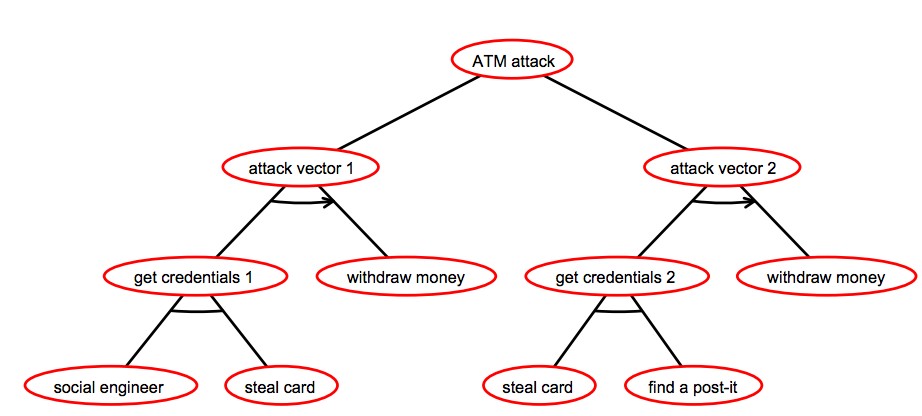}
  \end{center}
  \caption{Canonical Attack Tree for an ATM attack from Figure~2 on page 3 of Kordy et al.~\cite{Kordy2017}}
  \label{fig:atm-tree2}
\end{figure}
We can model this attack tree as a formula in ATLL as well (the base
attacks are the same):
\[
\begin{array}{lll}
  \ATLLnt{T_{{\mathrm{2}}}} := \ATLLsym{(}  \ATLLsym{(}  \ATLLnt{B_{{\mathrm{1}}}}  \odot  \ATLLnt{B_{{\mathrm{2}}}}  \ATLLsym{)}  \rhd  \ATLLnt{B_{{\mathrm{4}}}}  \ATLLsym{)}  \sqcup  \ATLLsym{(}  \ATLLsym{(}  \ATLLnt{B_{{\mathrm{1}}}}  \odot  \ATLLnt{B_{{\mathrm{3}}}}  \ATLLsym{)}  \rhd  \ATLLnt{B_{{\mathrm{4}}}}  \ATLLsym{)}
\end{array}
\]
The attack trees $\ATLLnt{T_{{\mathrm{1}}}}$ and $\ATLLnt{T_{{\mathrm{2}}}}$ represents the same attack, and
this can be proven in ATLL by showing that $\ATLLnt{T_{{\mathrm{1}}}}  \multimapboth  \ATLLnt{T_{{\mathrm{2}}}}$ where
$ \multimapboth $ denotes bi-implication.  The proof holds by using the
distributive rules for choice.

Modeling attack trees in linear logic has a number of benefits.
First, it connects attack trees back to logic, but in elegant and
simple way.  Kordy et al.~\cite{Kordy2010,Kordy:2012} proposed that
attack trees be modeled in propositional logic, but in this model
attacks can be freely duplicated and contracted which goes against the
process nature of an attack tree.  However, linear logic restores this
natural interpretation without loosing the process interpretation of
attack trees and without having to resort to complicated notation
unlike models similar to the situation calculus \cite{Samarji2013}.
By connecting attack trees to logic we can also tap into the long
standing research and development of automation, for example, SAT,
SMT, proof search, etc.  Finally, by connecting to logic we also
connect to the theory of statically typed functional programming
through the Curry-Howard-Lambek correspondence which will allow for
the development of a programming language that can be used to certify
the correctness of attack trees and the analysis performed using them.
In particular, we can use ATLL to design a functional scripting
language for the definition of attack trees that has the same
semantics as attack trees.

Before introducing ATLL we given several new logical models of attack
trees starting in
Section~\ref{sec:a_quaternary_semantics_for_sand_attack_trees} with a
very basic truth table semantics.  Then we lift this semantics into a
semantics of attack trees based on lineales in
Section~\ref{sec:lineale_semantics_for_sand_attack_trees}, but this
can be further lifted into a dialectica model which ATLL is based
(Section~\ref{sec:dialectica_semantics_of_sand_attack_trees}).

\textbf{Contributions.}  This paper offers the following
contributions:
\begin{itemize}
\item The first simple truth table semantics of SAND attack trees,
\item The first categorical model of attack trees in dialectica
  categories,
\item The Attack Tree Linear Logic (ATLL): a new intuitionistic linear
  logic for the specification and analysis of attack trees.

\item Horne et al.~\cite{horne2017semantics} propose to model attack
  trees in linear logic as well, but their logic is based on classical
  linear logic, and does not support full distribution of sequential
  conjunction over choice, while ATLL does.
\end{itemize}

Every section except
Section~\ref{sec:the_attack_tree_linear_logic_(atll)} has been
formalized in the Agda proof assistant\footnote{The Agda formalization
  can be found here
  \url{https://github.com/MonoidalAttackTrees/ATLL-Formalization}}.
Furthermore, all of the syntax used in this paper was formalized using
OTT~\cite{Sewell:2010}.

\section{SAND Attack Trees}
\label{sec:sand_attack_trees}
In this section we introduce SAND attack trees.  This formulation of
attack trees was first proposed by Jhawar et al. \cite{Jhawar:2015}.
\begin{definition}
  \label{def:atrees}
  Suppose $\mathsf{B}$ is a set of base attacks whose elements are
  denoted by $\ATmv{b}$.  Then an \textbf{attack tree} is defined by
  the following grammar:
  \[
  \begin{array}{lll}
    \ATnt{A},\ATnt{B},\ATnt{C},\ATnt{T} := \ATmv{b} \mid  \mathsf{OR}( \ATnt{A} , \ATnt{B} )  \mid  \mathsf{AND}( \ATnt{A} , \ATnt{B} )  \mid  \mathsf{SAND}( \ATnt{A} , \ATnt{B} ) \\
  \end{array}
  \]
  \noindent
  Equivalence of attack trees, denoted by $ \ATnt{A}  \approx  \ATnt{B} $, is defined as
  follows:
  \begin{center}
    \begin{math}
      \setlength{\arraycolsep}{10px}
      \begin{array}{lll}
        (\text{E}_1) &   \mathsf{OR}(  \mathsf{OR}( \ATnt{A} , \ATnt{B} )  , \ATnt{C} )   \approx   \mathsf{OR}( \ATnt{A} ,  \mathsf{OR}( \ATnt{B} , \ATnt{C} )  )  \\
        (\text{E}_2) &   \mathsf{AND}(  \mathsf{AND}( \ATnt{A} , \ATnt{B} )  , \ATnt{C} )   \approx   \mathsf{AND}( \ATnt{A} ,  \mathsf{AND}( \ATnt{B} , \ATnt{C} )  )  \\
        (\text{E}_3) &   \mathsf{SAND}(  \mathsf{SAND}( \ATnt{A} , \ATnt{B} )  , \ATnt{C} )   \approx   \mathsf{SAND}( \ATnt{A} ,  \mathsf{SAND}( \ATnt{B} , \ATnt{C} )  )  \\
        (\text{E}_4) &   \mathsf{OR}( \ATnt{A} , \ATnt{B} )   \approx   \mathsf{OR}( \ATnt{B} , \ATnt{A} )  \\
        (\text{E}_5) &   \mathsf{AND}( \ATnt{A} , \ATnt{B} )   \approx   \mathsf{AND}( \ATnt{B} , \ATnt{A} )  \\
        (\text{E}_6) &   \mathsf{AND}( \ATnt{A} ,  \mathsf{OR}( \ATnt{B} , \ATnt{C} )  )   \approx   \mathsf{OR}(  \mathsf{AND}( \ATnt{A} , \ATnt{B} )  ,  \mathsf{AND}( \ATnt{A} , \ATnt{C} )  )  \\
        (\text{E}_7) &   \mathsf{SAND}( \ATnt{A} ,  \mathsf{OR}( \ATnt{B} , \ATnt{C} )  )   \approx   \mathsf{OR}(  \mathsf{SAND}( \ATnt{A} , \ATnt{B} )  ,  \mathsf{SAND}( \ATnt{A} , \ATnt{C} )  )  \\
      \end{array}
    \end{math}
  \end{center}
\end{definition}
This definition of SAND attack trees differs slightly from Jhawar
et. al.'s~\cite{Jhawar:2015} definition.  They define $n$-ary
operators, but we only consider the binary case, because it fits better
with the models presented here and it does not loose any generality
because we can model the $n$-ary case using binary operators in the
obvious way.  Finally, they also include the equivalence $  \mathsf{OR}( \ATnt{A} , \ATnt{A} )   \approx  \ATnt{A} $, but it is not obvious how to include this in the models
presented here and we leave its addition to future work.



\section{A Quaternary Semantics for SAND Attack Trees}
\label{sec:a_quaternary_semantics_for_sand_attack_trees}
\newcommand{\forth}{\frac{1}{4}}
\newcommand{\half}{\frac{1}{2}}

Kordy et al.~\cite{Kordy:2012} gave a very elegant and simple
semantics of attack-defense trees in boolean algebras.  Unfortunately,
while their semantics is elegant it does not capture the resource
aspect of attack trees, it allows contraction, and it does not provide
a means to model sequential conjunction.  In this section we give a
semantics of attack trees in the spirit of Kordy et al.'s using a four
valued logic.

The propositional variables of our quaternary logic, denoted by $A$, $B$,
$C$, and $D$, range over the set $\mathsf{4} = \{0, \forth, \half,
1\}$.  We think of $0$ and $1$ as we usually do in boolean algebras,
but we think of $\forth$ and $\half$ as intermediate values that can
be used to break various structural rules.  In particular we will use
these values to prevent exchange for sequential conjunction from
holding, and contraction from holding for parallel and sequential
conjunction.
\begin{definition}
  \label{def:logical-connectives}
  The logical connectives of our four valued logic are defined as
  follows:
  \begin{itemize}
  \item[] Parallel and Sequential Conjunction:
    \begin{center}
      \begin{math}
        \setlength{\arraycolsep}{10px}
        \begin{array}{lll}
          \begin{array}{lll}
            A \odot_4 B = 1,\\
            \,\,\,\,\,\,\,\text{where neither $A$ nor $B$ are $0$}\\
          A \odot_4 B = 0, \text{otherwise}\\
          \\
        \end{array}
        &
        \begin{array}{lll}          
          A \rhd_4 B = 1,\\
          \,\,\,\,\,\,\,\text{where } A \in \{\half, 1\} \text{ and } B \neq 0\\
          \forth \rhd_4 B = \forth, \text{where $B \neq 0$}\\[2px]         
          A \rhd_4 B = 0, \text{otherwise}
        \end{array}
        \end{array}
      \end{math}
    \end{center}
  \item[] Choice: $A \sqcup_4 B = \mathsf{max}(A,B)$    
  \end{itemize}
\end{definition}
These definitions are carefully crafted to satisfy the necessary
properties to model attack trees.  Comparing these definitions with
Kordy et al.'s~\cite{Kordy:2012} work we can see that choice is
defined similarly, but parallel conjunction is not a product --
ordinary conjunction -- but rather a linear tensor product, and
sequential conjunction is not actually definable in a boolean algebra,
and hence, makes heavy use of the intermediate values to insure that
neither exchange nor contraction hold.  

We use the usual notion of equivalence between propositions, that is,
propositions $\phi$ and $\psi$ are considered equivalent, denoted by
$\phi \equiv \psi$, if and only if they have the same truth tables. In
order to model attack trees the previously defined logical connectives
must satisfy the appropriate equivalences corresponding to the
equations between attack trees.  These equivalences are all proven by
the following result.
\begin{lemma}[Properties of the Attack Tree Operators in the Quaternary Semantics]
  \label{lemma:props_atree_ops_quaternary-semantics}
  \begin{itemize}
  \item[] (Symmetry) For any $A$ and $B$, $A \bullet B \equiv B \bullet A$, for $\bullet \in \{\odot_4, \sqcup_4\}$.\\[-5px]
  \item[] (Symmetry for Sequential Conjunction) It is not the case that, for any $A$ and $B$, $A \rhd_4 B \equiv B \rhd_4 A$.\\[-5px]
  \item[] (Associativity) For any $A$, $B$, and $C$, $(A \bullet B) \bullet C \equiv A \bullet (B \bullet C)$, for $\bullet \in \{\odot_4, \rhd_4, \sqcup_4\}$.\\[-5px]
  \item[] (Contraction for Parallel and Sequential Conjunction) It is not the case that for any $A$, $A \bullet A \equiv A$, for $\bullet \in \{\odot_4, \rhd_4\}$.\\[-5px]
  \item[] (Distributive Law) For any $A$, $B$, and $C$, $A \bullet (B \sqcup_4 C) \equiv (A \bullet B) \sqcup_4 (A \bullet C)$, for $\bullet \in \{\odot_4, \rhd_4\}$.\\[-5px]
  \end{itemize}
\end{lemma}
\begin{proof}
  Symmetry, associativity, contraction for choice, and the
  distributive law for each operator hold by simply comparing truth
  tables.  As for contraction for parallel conjunction, suppose $A =
  \forth$.  Then by definition $A \odot_4 A = 1$, but $\forth$ is not
  $1$.  Contraction for sequential conjunction also fails, suppose $A
  = \half$.  Then by definition $A \rhd_4 A = 1$, but $\half$ is not
  $1$.  Similarly, symmetry fails for sequential conjunction. Suppose
  $A = \forth$ and $B = \half$.  Then $A \rhd_4 B = \forth$, but $B
  \rhd_4 A = 1$.
\end{proof}

At this point it is quite easy to model attack trees as formulas.  The
following defines their interpretation.
\begin{definition}
  \label{def:interp-aterms-quaternary}
  Suppose $\mathbb{B}$ is some set of base attacks, and $\nu :
  \mathbb{B} \mto \mathsf{PVar}$ is an assignment of base attacks to
  propositional variables.  Then we define the interpretation of
  $\mathsf{ATerms}$ to propositions as follows:
  \begin{center}
    \begin{math}
      \setlength{\arraycolsep}{5px}
      \begin{array}{lll}
        \begin{array}{lll}
          \interp{\ATermsmv{b} \in \mathbb{B}} & = & \nu(\ATermsmv{b})\\
          \interp{ \mathsf{AND}( \ATermsnt{A} , \ATermsnt{B} ) } & = & \interp{\ATermsnt{A}} \odot_4 \interp{\ATermsnt{B}}\\
        \end{array}
        &
        \begin{array}{lll}
          \interp{ \mathsf{SAND}( \ATermsnt{A} , \ATermsnt{B} ) } & = & \interp{\ATermsnt{A}} \rhd_4 \interp{\ATermsnt{B}}\\
          \interp{ \mathsf{OR}( \ATermsnt{A} , \ATermsnt{B} ) } & = & \interp{\ATermsnt{A}} \sqcup_4 \interp{\ATermsnt{B}}\\
        \end{array}
      \end{array}
    \end{math}
  \end{center}
\end{definition}
We can use this semantics to prove equivalences between attack trees.
\begin{lemma}[Equivalence of Attack Trees in the Quaternary Semantics]
  \label{lemma:equivalence_of_attack_trees}
  Suppose $\mathbb{B}$ is some set of base attacks, and $\nu :
  \mathbb{B} \mto \mathsf{PVar}$ is an assignment of base attacks to
  propositional variables.  Then for any attack trees $\ATermsnt{A}$ and
  $\ATermsnt{B}$, $ \ATermsnt{A}  \approx  \ATermsnt{B} $ if and only if $\interp{\ATermsnt{A}} \equiv \interp{\ATermsnt{B}}$.
\end{lemma}
\begin{proof}
  This proof holds by induction on the form of $ \ATermsnt{A}  \approx  \ATermsnt{B} $.
\end{proof}
This is a very simple and elegant semantics, but it also leads to a
more substantial theory.



\section{Lineale Semantics for SAND Attack Trees}
\label{sec:lineale_semantics_for_sand_attack_trees}
Classical natural deduction has a semantics in boolean algebras, and
so the semantics in the previous section begs the question of whether
there is a natural deduction system that can be used to reason about
attack trees.  We answer this question in the positive, but before
defining the logic we first build up a non-trivial concrete
categorical model of our desired logic in dialectica spaces, but this
first requires the abstraction of the quaternary semantics into a
preorder semantics we call the lineale semantics of SAND attack trees.
This semantics will live at the base of the dialectica space model
given in the next section, but it also begins to shed light on new and
interesting reasoning tools for attack trees.

We denote by $\leq_4 : \Four \times \Four \to \Four$ the obvious
preorder on $\Four$ making $(\Four, \leq_4)$ a preordered set
(proset).  It is well known that every preordered set induces a
category whose objects are the elements of the carrier set, here
$\Four$, and morphisms $\Hom{\Four}{a}{b} = a \leq_4 b$.  Composition
of morphisms hold by transitivity and identities exists by
reflexivity.  Under this setting it is straightforward to show that
for any propositions $\phi$ and $\psi$ over $\Four$ we have $\phi
\equiv \psi$ if and only if $\phi \leq_4 \psi$ and $\psi \leq_4 \phi$.
Thus, every result proven for the logical connectives on $\Four$ in
the previous section induce properties on morphisms in this setting.

In addition to the induced properties just mentioned we also have the
following new ones which are required when lifting this semantics to
dialectica spaces, but are also important when building a
corresponding logic.
\begin{lemma}[Functorality]
  \label{lemma:functorality}
  For any $A$, $B$, $C$, and $D$, if $A \leq_4 C$ and $B \le_4 D$,
  then $(A \bullet B) \leq_4 (C \bullet D)$, for $\bullet \in \{\odot_4,
  \rhd_4, \sqcup_4 \}$.
\end{lemma}
\begin{proof}
  Each part holds by case analysis over $A$, $B$, $C$,
  and $D$.  In any cases where $(A \bullet B) \leq_4 (C \bullet D)$
  does not hold, then one of the premises will also not hold.
\end{proof}
The logic we are building up is indeed intuitionistic, but none of the
operators we have introduced thus far are closed, but we can define
the standard symmetric linear tensor product in $\Four$ that is
closed.
\begin{definition}
  \label{def:tensor-and-implication}
  The following defines the linear tensor product on $\Four$ as well
  as linear implication:  
    \begin{center}
      \begin{math}
        \setlength{\arraycolsep}{10px}
        \begin{array}{lll}
          \begin{array}{lll}
            A \otimes_4 B = \mathsf{max}(A,B), \\
          \,\,\,\,\,\,\,\,\,\,\,\,\text{where $A$ nor $B$ are $0$}\\
          A \otimes_4 B = 0, \text{otherwise}
          \end{array}
          &
          \begin{array}{lll}
          A \limp_4 B = 0, \text{where $B <_4 A$}\\
          A \limp_4 A = A, \text{where $A \in \{\forth,\half\}$}\\
          A \limp_4 B = 1, \text{otherwise}
        \end{array}
        \end{array}
      \end{math}
    \end{center}
    The unit of the tensor product is $I_4 = \forth$.\\   
\end{definition}
The expected monoidal properties hold for the tensor product.
\begin{lemma}[Tensor is Symmetric Monoidal Closed]
  \label{lemma:tensor_is_symmetric_monoidal_closed}
  \begin{itemize}
  \item[] (Symmetry) For any $A$ and $B$, $A \otimes_4 B \equiv B \otimes A$.\\[-5px]
  \item[] (Associativity) For any $A$, $B$, and $C$, $(A \otimes_4 B) \otimes_4 C \equiv A \otimes_4 (B \otimes_4 C)$.\\[-5px]
  \item[] (Unitors) For any $A$, $(A \otimes I_4) \equiv A \equiv (I_4 \otimes A)$.\\[-5px]
  \item[] (Tensor is Functorial) For any $A$, $B$, $C$, and $D$, if $A \leq_4 C$ and $B \le_4 D$, then
    $(A \otimes_4 B) \leq_4 (C \otimes_4 D)$.\\[-5px]
  \item[] (Implication is Functorial) For any $A$, $B$, $C$, and $D$, if
    $C \leq_4 A$ and $B \leq_4 D$, then $(A \limp_4 B) \leq_4 (C
    \limp_4 D)$.\\[-5px]
  \item[] (Closure) For any $A$, $B$, and $C$, $(A \otimes_4 B) \leq_4 C$ if and only if $A \leq_4 (B \limp_4 C)$.
  \end{itemize}
\end{lemma}
\begin{proof}
  The top three cases hold by simply comparing truth tables. Finally,
  the last three cases hold by a case analysis over $A$, $B$, $C$, and
  $D$.  If at any time the conclusion is false, then one of the
  premises will also be false.
\end{proof}

We now define lineales which depend on the notion of a monoidal
proset.  The definition of lineales given here is a slight
generalization over the original definition given by Hyland and de
Paiva -- see Definition 1 of \cite{Hyland:1991}.  They base lineales
on posets instead of prosets, but the formalization given here shows
that anti-symmetry can be safely dropped.
\begin{definition}
  \label{def:monoidal-proset}
  A \textbf{monoidal proset} is a proset, $(L, \leq)$, with a given
  symmetric monoidal structure $(L, \circ, e)$.  That is, a set $L$
  with a given binary relation $\leq : L \times L \to L$ satisfying
  the following:
  \begin{itemize}
  \item (reflexivity) $a \leq a$ for any $a \in L$
  \item (transitivity) If $a \leq b$ and $b \leq c$, then $a \leq c$
  \end{itemize}
  together with a monoidal structure $(\circ, e)$ consisting of a
  binary operation, called multiplication, $\circ : L \times L \to L$
  and a distinguished element $e \in L$ called the unit such that the
  following hold:
  \begin{itemize}
  \item (associativity) $(a \circ b) \circ c = a \circ (b \circ c)$
  \item (identity) $a \circ e = a = e \circ a$
  \item (symmetry) $a \circ b = b \circ a$
  \end{itemize}
  Finally, the structures must be compatible, that is, if $a \leq b$,
  then $a \circ c \leq b \circ c$ for any $c \in L$.
\end{definition}
Now a lineale can be seen as essentially a symmetric monoidal closed
category in the category of prosets.
\begin{definition}
  \label{def:lineale}
  A \textbf{lineale} is a monoidal proset, $(L, \leq, \circ, e)$, with
  a given binary operation, called implication, $\limp : L \times L
  \to L$ such that the following hold:
  \begin{itemize}
  \item (relative complement) $(a \limp b) \circ a \leq b$ 
  \item (adjunction) If $a \circ y \leq b$, then $y \leq a \limp b$
  \end{itemize}
\end{definition}
The set $\mathsf{2} = \{0,1\}$ is an example of a lineale where the
order is the usual one, the multiplication is boolean conjunction, and
the implication is boolean implication.  This example is not that
interesting, because $\mathsf{2}$ is a boolean algebra.  An example of
a proper lineale can be given using the three element set
$\mathsf{3} = \{0,\half, 1\}$, but one must be careful when
defining lineales, because it is possible to instead define Heyting
algebras, and hence, become nonlinear.

Given the operations and properties shown for $(\Four, \leq_4)$ above
we can easily prove that $(\Four, \leq_4)$ defines a lineale.
\begin{lemma}
  \label{lemma:four_is_a_lineale}
  The proset, $(\Four, \leq_4, \otimes_4, I_4, \limp_4)$ is a lineale.
\end{lemma}
\begin{proof}
  First, $(\Four, \leq_4, \otimes_4, I_4)$ defines a monoidal proset,
  because the tensor product is associative, $I_4$ is the identity,
  and symmetric by
  Lemma~\ref{lemma:tensor_is_symmetric_monoidal_closed}.  We can also
  show that the tensor product is compatible, that is, if $A \leq_4
  B$, then $(A \otimes_4 C) \leq_4 (B \otimes C)$ for any $C$.
  Suppose $A \leq_4 B$, then by reflexivity we also know that $C
  \leq_4 C$.  Thus, by functorality,
  Lemma~\ref{lemma:tensor_is_symmetric_monoidal_closed}, we obtain our
  result.

  Finally, we show that $(\Four, \leq_4, \otimes_4, I_4, \limp_4)$ is
  a lineale.  The adjunction property already holds by
  Lemma~\ref{lemma:tensor_is_symmetric_monoidal_closed}, thus, all
  that is left to show is that the relative complement holds. We know
  by Lemma~\ref{lemma:tensor_is_symmetric_monoidal_closed} that for
  any $A$, $B$, and $C$, if $A \leq_4 (B \limp_4 C)$, then $(A
  \otimes_4 B) \leq_4 C$.  In addition, we know by reflexivity that
  $(A \limp_4 B) \leq_4 (A \limp_4 B)$, thus by the previous property we obtain
  that $((A \limp_4 B) \otimes_4 A) \leq_4 B$.
\end{proof}

The interpretation of attack trees into the lineale $(\Four, \leq_4,
\otimes_4, I_4, \limp_4)$ does not change from
Definition~\ref{def:interp-aterms-ternary}, but the equivalences
between attack trees, Lemma~\ref{lemma:equivalence_of_attack_trees},
can be abstracted.
\begin{lemma}[Equivalence of Attack Trees in the Lineale Semantics]
  \label{lemma:equivalence_of_attack_trees_lineale}
  Suppose $\mathbb{B}$ is some set of base attacks, and $\alpha :
  \mathbb{B} \mto \mathsf{PVar}$ is an assignment of base attacks to
  propositional variables.  Then for any attack trees $\ATermsnt{T_{{\mathrm{1}}}}$ and
  $\ATermsnt{T_{{\mathrm{2}}}}$, $ \ATermsnt{T_{{\mathrm{1}}}}  \approx  \ATermsnt{T_{{\mathrm{2}}}} $ if and only if $\interp{\ATermsnt{T_{{\mathrm{1}}}}} \leq_4 \interp{\ATermsnt{T_{{\mathrm{2}}}}}$ and
  $\interp{\ATermsnt{T_{{\mathrm{2}}}}} \leq_4 \interp{\ATermsnt{T_{{\mathrm{1}}}}}$.
\end{lemma}
\begin{proof}
  This proof holds by induction on the form of $ \ATermsnt{T_{{\mathrm{1}}}}  \approx  \ATermsnt{T_{{\mathrm{2}}}} $.
\end{proof}
This result seems basic, but has some interesting consequences.  It
implies that the notion of attack tree equivalence can be broken up
unto left-to-right and right-to-left implications which can themselves
be used to reason about properties of attack trees like when one tree
is a subtree of another.

In addition, this also implies that categorical models, and
equivalently by the Curry-Howard-Lambek Correspondence, logical models
of attack trees can support different notions of equivalence, because
equivalence of attack trees can be broken down into morphisms.  In
fact, in the next section we will lift the lineale semantics up into a
dialectica model, but dialectica models are models of linear logic.

Finally, the results of this section lead us to a more logical
viewpoint.  If we know $\interp{\ATermsnt{T_{{\mathrm{1}}}}} \leq_4 \interp{\ATermsnt{T_{{\mathrm{2}}}}}$, then
by closure $I_4 \leq_4 (\interp{\ATermsnt{T_{{\mathrm{1}}}}} \limp_4 \interp{\ATermsnt{T_{{\mathrm{2}}}}})$.
Thus, two attack trees are then equivalent if and only if they are
bi-conditionally related, i.e. $I_4 \leq_4 (\interp{\ATermsnt{T_{{\mathrm{1}}}}} \limp_4
\interp{\ATermsnt{T_{{\mathrm{2}}}}})$ and $I_4 \leq_4 (\interp{\ATermsnt{T_{{\mathrm{2}}}}} \limp_4
\interp{\ATermsnt{T_{{\mathrm{1}}}}})$.  Therefore, if we are able to find a logic that is
sound with respect to the semantics laid out thus far, then we can use
it to reason about attack trees using linear implication.



\section{Dialectica Semantics of SAND Attack Trees}
\label{sec:dialectica_semantics_of_sand_attack_trees}
In her thesis de Paiva \cite{dePaiva:1987} gave one of the first sound
and complete categorical models, called dialectica categories, of full
intuitionistic linear logic.  Her models arose from giving a
categorical definition to G\"odel's Dialectica interpretation.  de
Paiva defines a particular class of dialectica categories called $GC$
over a base category $C$, see page 41 of \cite{dePaiva:1987}.  She
later showed that by instantiating $C$ to $\mathsf{Sets}$, the
category of sets and total functions, that one arrives at concrete
instantiation of $GC$ she called $\dial{2}$ whose objects are called
\emph{dialectica spaces}, and then she abstracts $\dial{2}$ into a
family of concrete dialectica spaces, $\dial{\text{$L$}}$, by
replacing $\mathsf{2}$ with an arbitrary lineale $L$.

In this section we construct the dialectica category, $\dial{4}$, and
show that it is a model of attack trees.  This will be done by
essentially lifting each of the attack tree operators defined for the
lineale semantics given in the previous section into the dialectica
category.  Working with dialectica categories can be very complex due
to the nature of how they are constructed.  In fact, they are one of
the few examples of theories that are easier to work with in a proof
assistant than outside of one.  Thus, throughout this section we only
give brief proof sketches, but the interested reader will find the
complete proofs in the formalization.

We begin with the basic definition of $\dial{4}$, and prove it is a
category.
\begin{definition}
  \label{def:dialectica-model}
  The category of dialectica spaces over $\Four$, denoted by
  $\dial{4}$, is defined by the following data:
  \begin{itemize}
  \item objects, or dialectica spaces, are triples $(U, X, \alpha)$
    where $U$ and $X$ are sets, and $\alpha : U \to X \to \Four$ is a
    relation on $\Four$.

  \item morphisms are pairs $(f,F) : (U,X,\alpha) \to (V,Y,\beta)$
    where $f : U \to V$ and $F : Y \to X$ such that for any $u \in U$
    and $y \in Y$, $\alpha(u,F(y)) \leq_4 \beta(f(u),y)$.
  \end{itemize}
\end{definition}

\begin{lemma}
  \label{lemma:dial4_is_a_category}
  The structure $\dial{4}$ is a category.
\end{lemma}
\begin{proof}
  Identity morphisms are defined by $(\id_U, \id_X) : (U , X , \alpha)
  \to (U , X, \alpha)$, and the property on morphisms holds by
  reflexivity.  Given two morphism $(f, F) : (U , X, \alpha) \to (V ,
  Y , \beta)$ and $(g, G) : (V , Y, \beta) \to (W , Z, \gamma)$, then
  their composition is defined by $(f;g, G;F) $ $: (U , X , \alpha) \to
  (W , Z , \gamma)$ whose property holds by transitivity.  Proving
  that composition is associative and respects identities is
  straightforward.
\end{proof}

Next we show that $\dial{4}$ is symmetric monoidal closed.  The
definitions of both the tensor product and the internal hom will be
defined in terms of their respective counterparts in the lineale
semantics.
\begin{definition}
  \label{def:dialectica-model-smcc}
  The following defines the tensor product and the internal hom:
  \begin{itemize}
  \item[] (Tensor Product) Suppose $A = (U , X , \alpha)$ and $B = (V, Y, \beta)$, then
    define $A \otimes B = (U \times V, (V \to X) \times (U \to Y), \alpha \otimes_r \beta)$, where
    $(\alpha \otimes_r \beta)(u, v)(f, g) = (\alpha\,u\,(f v)) \otimes_4 (\beta\,v\,(g\,u))$.\\[-5px]

  \item[] (Internal Hom) Suppose $A = (U , X , \alpha)$ and $B = (V, Y, \beta)$, then
    define $A \limp B = ((U \to V) \times (Y \to X), U \times Y, \alpha \limp_r \beta)$, where
    $(\alpha \limp_4 \beta)(f , g)(u , y) = (\alpha\,u\,(g\,y)) \limp_4 (\beta\,(f\,u)\,y)$.
  \end{itemize}
  The unit of the tensor product is defined by $I = (\top, \top,
  (\lambda x.\lambda y.I_4))$, where $\top$ is the final object in
  $\mathsf{Set}$.
\end{definition}
The following properties hold for the previous constructions.
\begin{lemma}[SMCC Properties for $\dial{4}$]
  \label{lemma:smcc_properties_for_dial4}
  \begin{itemize}
  \item[] (Functorality for Tensor) Given morphisms $f : A \mto C$ and
    $g : B \mto D$, then there is a morphism $f \otimes g : (A \otimes B) \mto (C \otimes D)$.\\[-5px]
  \item[] (Associator) There is a natural isomorphism, $\alpha_{A,B,C} : (A \otimes B) \otimes C \mto A \otimes (B \otimes C)$.\\[-5px]
  \item[] (Unitors) There are natural isomorphisms, $\lambda_A : (I \otimes A) \mto A$ and $\rho_A : (A \otimes I) \mto A$.\\[-5px]
  \item[] (Symmetry) There is a natural transformation, $\beta_{A,B} : (A \otimes B) \mto (B \otimes A)$ that is involutive.\\[-5px]
  \item[] (Functorality for the Internal Hom) Given morphism $f : C \mto A$ and $g : B \mto D$, then there is a morphism $f \limp g : (A \limp B) \mto (C \limp D)$.\\[-5px]
  \item[] (Adjunction) There is a natural bijection:\vspace{-5px}
    \[ \mathsf{curry} : \Hom{\dial{4}}{A \otimes B}{C} \cong \Hom{\dial{4}}{A}{B \limp C}. \]
  \end{itemize}
  Finally, the coherence diagrams for symmetric monoidal categories --
  which we omit to conserve space, but can be found here
  \cite{MacLane:1971} -- also hold for the natural transformations
  above.
\end{lemma}
\begin{proof}
  These properties are not new, and their proofs follow almost exactly
  de Paiva's proofs from her thesis \cite{dePaiva:1987}.  The complete
  proofs for each of the cases above, including the proofs for the
  symmetric monoidal coherence diagrams, can be found in the
  formalization.
\end{proof}

The constructions on $\dial{4}$ given so far are not new, but the
constructions for the attack tree operators for parallel conjunction,
sequential conjunction, and choice are new to dialectica categories,
but it turns out that the definition of choice we give here has been
previously used in a different categorical construction called the
category of Chu spaces.
\begin{definition}
  \label{def:attack-tree-ops-dialectica}
  The attack tree operators are defined in $\dial{4}$ as follows:
  \begin{itemize}
  \item[] (Parallel Conjunction) Suppose $A = (U , X , \alpha)$ and $B = (V , Y , \beta)$, then
    $A \odot B = (U \times V, X \times Y, \alpha \odot_r \beta)$, where
    $(\alpha \odot_r \beta)(u , v)(x , y) = (\alpha\,u\,x) \odot_4 (\beta\,v\,y)$.\\[-5px]

  \item[] (Sequential Conjunction) $A = (U , X , \alpha)$ and $B = (V , Y , \beta)$, then
    $A \rhd B = (U \times V, X \times Y, \alpha \rhd_r \beta)$, where
    $(\alpha \rhd_r \beta)(u , v)(x , y) = (\alpha\,u\,x) \rhd_4 (\beta\,v\,y)$.\\[-5px]
    
  \item[] (Choice) $A = (U , X , \alpha)$ and $B = (V , Y , \beta)$, then
    $A \sqcup B = (U + V, X + Y, \alpha \sqcup_r \beta)$, where
    \begin{center}
      \begin{math}
        \begin{array}{lll}
          (\alpha \odot_r \beta)\,a\,b = \alpha\,a\,b, \text{when $a \in U$ and $b \in X$}\\
          (\alpha \odot_r \beta)\,a\,b = \beta\,a\,b, \text{when $a \in V$ and $b \in Y$}\\
          (\alpha \odot_r \beta)\,a\,b = 0, \text{otherwise}\\
        \end{array}
      \end{math}
    \end{center}
  \end{itemize}
\end{definition}
The definitions of parallel and sequential conjunction are quite
literally the lifting of their lineale counterparts.  The parallel and
sequential operators on $(\mathsf{4}, \leq_4,\otimes_4,I_4,\limp_4)$,
$\odot_4$ and $\rhd_4$, restrict the cartesian product to the required
properties for attack trees.  Now choice must be carefully constructed
so that we may prove the required distributive law.

Given a dialectica space, $(U, X, \alpha)$, we can consider $U$ as a
set of actions and $X$ as a set of states.  Then given an action, $a
\in U$, and a state, $q \in X$, $\alpha\,a\,q$, indicates whether
action $a$ will execute in state $q$.  This implies that an action $a$
and a state $q$ of $A \sqcup B$, for $A = (U , X , \alpha)$ and $B =
(V , Y , \beta)$, are either an action of $A$ or an action of $B$, and
a state of $A$ or a state of $B$.  Then an action, $a$, of $A \sqcup
B$ will execute in state $q$ of $A \sqcup B$ if they are both from $A$
or both from $B$.  Thus, the definition of choice very much fits the
semantics of a choice operator.  It is well known that the cartesian
product distributes over the disjoint union in $\mathsf{Sets}$, and
because of the definitions of parallel and sequential conjunction, and
choice, these properties lift up into $\dial{4}$.

It turns out that the definition of choice given here is not new at
all, but first appeared as the choice operator used for modeling
concurrency in Chu spaces due to Gupta and Pratt \cite{Gupta:1994}.
Chu spaces are the concrete objects of Chu categories just like
dialectica spaces are the concrete objects of dialectica categories.
In fact, Chu categories and dialectica categories are cousins
\cite{dePaiva:2006b}.  Chu and dialectica categories have exactly the
same objects, but the condition on morphisms is slightly different,
for Chu categories the condition uses equality instead of the
preorder.  The impact of this is significant, Chu spaces are a model
of classical linear logic, while dialectica categories are a model of
intuitionistic linear logic.  

The following gives all of the properties that hold for the attack
tree operators in $\dial{4}$.
\begin{lemma}[Properties of the Attack Tree Operators in $\dial{4}$]
  \label{lemma:properties_of_the_attack_tree_operators_in_dial4}
  \begin{itemize}
  \item[] (Functorality) Given morphisms $f : A \mto C$ and
    $g : B \mto D$, then there is a morphism $f \bullet g : (A \bullet B) \mto (C \bullet D)$, for $\bullet \in \{\odot, \rhd, \sqcup\}$.\\[-5px]
  \item[] (Associativity) There is a natural isomorphism, $\alpha^\bullet_{A,B,C} : (A \bullet B) \bullet C \mto A \bullet (B \bullet C)$, for $\bullet \in \{\odot, \rhd, \sqcup\}$.\\[-5px]
  \item[] (Symmetry) There is a natural transformation, $\beta^\bullet_{A,B} : (A \bullet B) \mto (B \bullet A)$ that is involutive, for $\bullet \in \{\odot, \rhd, \sqcup\}$.\\[-5px]
  \item[] (Distributive Law) There is a natural isomorphism, $distl^\bullet : A \bullet (B \sqcup C) \mto (A \bullet B) \sqcup (A \bullet C)$, for $\bullet \in \{\odot, \rhd\}$.\\[-5px]
  \end{itemize}
\end{lemma}

At this point we can interpret attack trees into $\dial{4}$.
\begin{definition}
  \label{def:interp-aterms-ternary}
  Suppose $\mathbb{B}$ is some set of base attacks, and $\nu :
  \mathbb{B} \mto \obj{\dial{4}}$ is an assignment of base attacks to
  dialectica spaces.  Then we define the interpretation of
  $\mathsf{ATerms}$ to objects of $\dial{4}$ as follows:
  \begin{center}
    \begin{math}
      \setlength{\arraycolsep}{5px}
      \begin{array}{lll}
        \begin{array}{rll}
          \interp{\ATermsmv{b} \in \mathbb{B}} & = & \nu(\ATermsmv{b})\\
        \interp{ \mathsf{AND}( \ATermsnt{A} , \ATermsnt{B} ) } & = & \interp{\ATermsnt{A}} \odot \interp{\ATermsnt{B}}\\
        \end{array}
        &
        \begin{array}{rll}
          \interp{ \mathsf{OR}( \ATermsnt{A} , \ATermsnt{B} ) } & = & \interp{\ATermsnt{A}} \sqcup \interp{\ATermsnt{B}}\\
          \interp{ \mathsf{SAND}( \ATermsnt{A} , \ATermsnt{B} ) } & = & \interp{\ATermsnt{A}} \rhd \interp{\ATermsnt{B}}\\          
        \end{array}
      \end{array}
    \end{math}
  \end{center}
\end{definition}
Then we have the following result.
\begin{lemma}[Equivalence of Attack Trees in the Dialectica Semantics]
  \label{lemma:equivalence_of_attack_trees_lineale}
  Suppose $\mathbb{B}$ is some set of base attacks, and $\nu :
  \mathbb{B} \mto \obj{\dial{4}}$ is an assignment of base attacks to
  dialectica spaces.  Then for any attack trees $\ATermsnt{A}$ and $\ATermsnt{B}$,
  $ \ATermsnt{A}  \approx  \ATermsnt{B} $ if and only if there is a natural isomorphism $m :
  \interp{\ATermsnt{A}} \mto \interp{\ATermsnt{B}}$ in $\dial{4}$.
\end{lemma}
\begin{proof}
  This proof holds by induction on the form of $ \ATermsnt{A}  \approx  \ATermsnt{B} $.
\end{proof}



\section{The Attack Tree Linear Logic (ATLL)}
\label{sec:the_attack_tree_linear_logic_(atll)}
In this section we take what we have learned by constructing the
dialectica model and define a intuitionistic linear logic, called the
attack tree linear logic (ATLL), that can be used to prove
equivalences between attack trees as linear implications.  ATLL is
based on the logic of bunched implications (BI) \cite{Ohearn:2003}, in
that, contexts will be trees.  This is necessary to be able to include
parallel and sequential conjunction, and choice within the same logic,
because they all have different structural rules associated with them.

The syntax for formulas and contexts are defined by the following
grammar.
\begin{center}
  \begin{math}
    \begin{array}{lll}
      \ATLLnt{A},\ATLLnt{B},\ATLLnt{C},D,\ATLLnt{T} := \ATLLmv{N} \mid \ATLLnt{A}  \sqcup  \ATLLnt{B} \mid \ATLLnt{A}  \odot  \ATLLnt{B} \mid \ATLLnt{A}  \rhd  \ATLLnt{B} \mid \ATLLnt{A}  \multimap  \ATLLnt{B}\\
      \Gamma,\Delta := \ATLLsym{*} \mid \ATLLnt{A} \mid \Gamma  \ATLLsym{,}  \Delta \mid \Gamma  \ATLLsym{;}  \Delta \mid
       \Gamma  \bullet  \Delta \\
    \end{array}
  \end{math}
\end{center}
ATLL formulas are not surprising, but we denote base attacks by atomic
formulas represented here by $\ATLLmv{N}$.  The syntax for contexts are
similar to the contexts in BI.  Contexts are trees with three types of
nodes denoted by $\Gamma  \ATLLsym{,}  \Delta$ for parallel conjunction, $\Gamma  \ATLLsym{;}  \Delta$ for
sequential conjunction, and $ \Gamma  \bullet  \Delta $ for choice.  They all have
units, but we overload the symbol $\ATLLsym{*}$ to represent them all.

The ATLL inference rules are given in Figure~\ref{fig:atll-rules}.
\begin{figure}
  \begin{mdframed}
    \begin{mathpar}
      \ATLLdruleLXXvar{} \and
      \ATLLdruleLXXnode{} \and
      \ATLLdruleLXXCtx{} \and
      \ATLLdruleLXXparaI{} \and
      \ATLLdruleLXXchoiceI{} \and
      \ATLLdruleLXXseqI{} \and
      \ATLLdruleLXXparaE{} \and
      \ATLLdruleLXXchoiceE{} \and
      \ATLLdruleLXXseqE{} \and
      \ATLLdruleLXXlimpI{} \and
      \ATLLdruleLXXlimpE{}
    \end{mathpar}
  \end{mdframed}
  \caption{ATLL Inference Rules}
  \label{fig:atll-rules}
\end{figure}
The inference rules are fairly straightforward.  We denote by
$ \Delta ( \Gamma ) $ the context $\Delta$ with a subtree -- subcontext --
$\Gamma$.  This syntax is used to modify the context across inference
rules.

Perhaps the most interesting rule is the $\text{CM}$ rule which stands
for context morphism.  This rule is a conversion rule for manipulation
of the context.  It depends on a judgment $ \Gamma_{{\mathrm{1}}}  \vdash  \Gamma_{{\mathrm{2}}} $ which can be
read as the context $\Gamma_{{\mathrm{1}}}$ can be transformed into the context
$\Gamma_{{\mathrm{2}}}$.  This judgment is defined by the rules in
Figure~\ref{fig:atll-cm-rules}.
\begin{figure}
  \begin{mdframed}
    \begin{mathpar}
      \ATLLdruleCXXid{} \and
      \ATLLdruleCXXc{} \and
      \ATLLdruleCXXaOne{} \and
      \ATLLdruleCXXuOne{} \and
      \ATLLdruleCXXuTwo{} \and
      \ATLLdruleCXXeOne{} \and
      \ATLLdruleCXXeTwo{} \and
      \ATLLdruleCXXdOne{} \and
      \ATLLdruleCXXdTwo{} \and
      \ATLLdruleCXXdThree{} \and
      \ATLLdruleCXXdFour{} \and
    \end{mathpar}
  \end{mdframed}
  \caption{Context Morphisms}
  \label{fig:atll-cm-rules}
\end{figure}
Context morphisms are designed to induce structural rules for some of
the logical connectives and not for others.  For example, parallel
conjunction and choice should be commutative, but sequential
conjunction should not be.  The rules for associativity and the unit
rules mention the operator, $\circ$, this operator ranges over '$,$',
'$;$', and '$\bullet$'.

One interesting, and novel aspect of this logic in contrast to BI is
we can use the context morphisms to induce distributive laws between
the various tensor products.  The rules $\text{dist}_1$ and
$\text{dist}_2$ induce the property that sequential conjunction
distributes over choice, and $\text{dist}_3$ and $\text{dist}_4$
induce the property that parallel conjunction distributes over choice.

The interpretation of attack trees as ATLL formulas is obvious at this
point where base attacks are atomic formulas and we denote this
interpretation as $\interp{\ATLLnt{T}}$ for some attack tree $\ATLLnt{T}$.  The
most interesting part about this interpretation is that we can now use
linear implication to prove properties about attack trees.  First, we
can derive all of the required equivalences in ATLL.
\begin{lemma}[Attack Tree Logical Equivalences]
  \label{lemma:attack_tree_logical_equivalences}
  The following hold for any ATLL formulas $\ATLLnt{A}$, $\ATLLnt{B}$, and $\ATLLnt{C}$.
  \begin{itemize}
  \item $ \ATLLsym{*}  \vdash  \ATLLsym{(}  \ATLLnt{A}  \sqcup  \ATLLnt{B}  \ATLLsym{)}  \multimapboth  \ATLLsym{(}  \ATLLnt{B}  \sqcup  \ATLLnt{A}  \ATLLsym{)} $
  \item $ \ATLLsym{*}  \vdash  \ATLLsym{(}  \ATLLnt{A}  \odot  \ATLLnt{B}  \ATLLsym{)}  \multimapboth  \ATLLsym{(}  \ATLLnt{B}  \odot  \ATLLnt{A}  \ATLLsym{)} $
  \item $ \ATLLsym{*}  \vdash  \ATLLsym{(}  \ATLLsym{(}  \ATLLnt{A}  \sqcup  \ATLLnt{B}  \ATLLsym{)}  \sqcup  \ATLLnt{C}  \ATLLsym{)}  \multimapboth  \ATLLsym{(}  \ATLLnt{A}  \sqcup  \ATLLsym{(}  \ATLLnt{B}  \sqcup  \ATLLnt{C}  \ATLLsym{)}  \ATLLsym{)} $
  \item $ \ATLLsym{*}  \vdash  \ATLLsym{(}  \ATLLsym{(}  \ATLLnt{A}  \odot  \ATLLnt{B}  \ATLLsym{)}  \odot  \ATLLnt{C}  \ATLLsym{)}  \multimapboth  \ATLLsym{(}  \ATLLnt{A}  \odot  \ATLLsym{(}  \ATLLnt{B}  \odot  \ATLLnt{C}  \ATLLsym{)}  \ATLLsym{)} $
  \item $ \ATLLsym{*}  \vdash  \ATLLsym{(}  \ATLLsym{(}  \ATLLnt{A}  \rhd  \ATLLnt{B}  \ATLLsym{)}  \rhd  \ATLLnt{C}  \ATLLsym{)}  \multimapboth  \ATLLsym{(}  \ATLLnt{A}  \rhd  \ATLLsym{(}  \ATLLnt{B}  \rhd  \ATLLnt{C}  \ATLLsym{)}  \ATLLsym{)} $
  \item $ \ATLLsym{*}  \vdash  \ATLLsym{(}  \ATLLnt{A}  \odot  \ATLLsym{(}  \ATLLnt{B}  \sqcup  \ATLLnt{C}  \ATLLsym{)}  \ATLLsym{)}  \multimapboth  \ATLLsym{(}  \ATLLsym{(}  \ATLLnt{A}  \odot  \ATLLnt{B}  \ATLLsym{)}  \sqcup  \ATLLsym{(}  \ATLLnt{A}  \odot  \ATLLnt{C}  \ATLLsym{)}  \ATLLsym{)} $
  \item $ \ATLLsym{*}  \vdash  \ATLLsym{(}  \ATLLnt{A}  \rhd  \ATLLsym{(}  \ATLLnt{B}  \sqcup  \ATLLnt{C}  \ATLLsym{)}  \ATLLsym{)}  \multimapboth  \ATLLsym{(}  \ATLLsym{(}  \ATLLnt{A}  \rhd  \ATLLnt{B}  \ATLLsym{)}  \sqcup  \ATLLsym{(}  \ATLLnt{A}  \rhd  \ATLLnt{C}  \ATLLsym{)}  \ATLLsym{)} $
  \end{itemize}
\end{lemma}
Using the previous lemma we can now completely reason about
equivalences of attack trees in ATLL.  Another important aspect of
ATLL is that we can use either the left-to-right directions or the
right-to-left directions of the previous bi-implications to simplify
attack trees into normal forms.  In addition, the logical
interpretation leads to new and interesting questions, for example,
adding additional structural rules, like weakening, could also open
the door for proving when one attack tree is a subattack tree of
another.  This concept is yet to appear in the literature, but has
practical applications.


\section{Related and Future Work}
\label{sec:related_and_future_work}
\textbf{Related Work.}  Horne et al.~\cite{horne2017semantics} also
propose modeling SAND attack trees using linear logic, but they base
their work on pomsets and classical linear logic.  In addition, their
logic cannot derive the distributive law for sequential conjunction up
to an equivalence, but they can derive $ \ATLLsym{*}  \vdash  \ATLLsym{(}  \ATLLsym{(}  \ATLLnt{A}  \rhd  \ATLLnt{B}  \ATLLsym{)}  \sqcup  \ATLLsym{(}  \ATLLnt{A}  \rhd  \ATLLnt{C}  \ATLLsym{)}  \ATLLsym{)}  \multimap  \ATLLsym{(}  \ATLLnt{A}  \rhd  \ATLLsym{(}  \ATLLnt{B}  \sqcup  \ATLLnt{C}  \ATLLsym{)}  \ATLLsym{)} $.  The full equivalence is derivable in ATLL
however.

The logic of bunched implications \cite{Ohearn:2003} has already been
shown to be able to support non-commutative operators by O'hearn, but
here we show how distributive laws can be controlled using properties
on contexts.

de Paiva~\cite{depaiva1991} shows how to model non-commutative
operators in dialectica categories, but here we show an alternative
way of doing this, and we extend the model to include more operators
like choice and its distributive laws.

\textbf{Future Work.}  We plan to build a term assignment for ATLL
that can be used a scripting language for defining and reasoning about
attack trees.  In addition, we plan to extend equivalence of attack
trees with contraction for choice, and to investigate adding a
modality that adds weakening to ATLL, and then, this modality could be
used to reason about subattack trees.  Finally, we leave the proof
theory of ATLL to future work.

\bibliographystyle{plain}






\end{document}